\documentclass[runningheads]{llncs}

\usepackage[T1]{fontenc}
\usepackage{parskip}
\usepackage{amsmath}
\usepackage{amssymb}
\usepackage{mathtools}
\usepackage{thmtools}
\usepackage{bm}

\usepackage[pdfpagelabels,bookmarks=false]{hyperref}
\hypersetup{colorlinks, linkcolor=darkblue, citecolor=darkgreen, urlcolor=darkblue}
\usepackage[capitalize, nameinlink]{cleveref}
\usepackage{enumitem}

\usepackage{float}
\usepackage{caption}
\usepackage{subcaption}

\usepackage{tikz}
\usetikzlibrary{positioning}
\usetikzlibrary{calc}
\usetikzlibrary{decorations.markings}
\usetikzlibrary{decorations.pathreplacing}
\usetikzlibrary{arrows.meta}
\usetikzlibrary{shapes.geometric}

\foreach \x in {A,...,Z}{
	\expandafter\xdef\csname b\x\endcsname{\noexpand\mathbb{\x}}
	\expandafter\xdef\csname c\x\endcsname{\noexpand\mathcal{\x}}
	\expandafter\xdef\csname f\x\endcsname{\noexpand\mathfrak{\x}}
}

\renewcommand{\bar}[1]{\overline{#1}}
\renewcommand{\hat}[1]{\widehat{#1}}

\newcommand{\abs}[1]{\left\lvert#1\right\rvert}

\usepackage{xurl} 
\usepackage[backend=bibtex, maxbibnames=10]{biblatex}
\usepackage{setspace}
\usepackage{colortbl}
\usepackage{tcolorbox}
\usepackage{thm-restate}

\newcommand{\Pminexcess}{\textsc{MinExcess}}
\newcommand{\Psubspaceavoidingminexcess}{\textsc{LSA-MinExcess}}
\newcommand{\Pnonzerominexcess}{\textsc{NZ-MinExcess}}

\definecolor{darkblue}{rgb}{0,0,0.38}
\definecolor{darkred}{rgb}{0.8,0,0}
\definecolor{darkgreen}{rgb}{0.1,0.35,0}
\definecolor{lightblue}{rgb}{0.35,0.6,0.8}

\begin{document}
\title{Nucleolus Computation by\\ Non-Zero-Constrained Optimization}
\author{Daniel Ebert and Antonia Ellerbrock}
\institute{Research Institute for Discrete Mathematics, Bonn, Germany\\
	\email{\{ebert,ellerbrock\}@dm.uni-bonn.de}}

\maketitle

\begin{abstract}
	We extend the list of games where the nucleolus is computable in polynomial time.
	Based on the classical MPS scheme, nucleolus computation can be reduced to 
	the problem of finding a coalition with minimum excess that does not belong to a given linear subspace.
	We call this problem \Psubspaceavoidingminexcess{}, 
	and show that it is equivalent to \Pnonzerominexcess{}:
	Given integral values per player, find a coalition with minimum excess whose player values do not sum up to $0$.
	Exploiting this representation, we prove that the nucleolus is computable in polynomial time for arboricity games, network strength games,
	and certain $b$-matching games.
	Along these lines, we show that for $b$-matching games with $b \leq 2$, \Psubspaceavoidingminexcess{} is polynomially equivalent to \textsc{Shortest Non-Zero Cycle}.

	Further, we prove that in general, linear subspace avoidance strictly increases the complexity of the minimum excess problem, even for monotone games.
	We still provide a reduction that trades linear subspace avoidance against an arbitrarily small approximation error.
	Finally, we show that the nucleolus is unstable in the following sense:
	A small change in the value function of the game can lead to a
	change in the nucleolus that is exponential in the number of players.
\end{abstract}

\section{Introduction}

A \emph{cooperative game} is a pair $(P,v)$ of a finite set of
\emph{players} $P$ and a value function $v: 2^P\to\bR$ on \emph{coalitions} $S \subseteq P$ with $v(\emptyset)=0$. $(P,v)$ is called
\emph{monotone} if \mbox{$v(S)\leq v(T)$} for all $S\subseteq T \subseteq P$,
and \emph{superadditive} if $v(S \cup T)\geq v(S) + v(T)$ for all $S,T\subseteq P$.
Usually, $v$ is not given explicitly, but by some combinatorial optimization problem, e.g., as the maximum weight of a matching in the subgraph $G[S]$ of some weighted graph $G$ for $S \subseteq P$ \cite{deng1999algorithmic}.
This way, the input size of $(P,v)$ can be controlled.
Given a cooperative game $(P,v)$, a \emph{cost allocation} is a vector $y\in\bR^P$.
We will write \mbox{$y(S) \coloneqq \sum_{p \in S} y_p$} for $S \subseteq P$. 
The \emph{nucleolus} by \textcite{schmeidler1969nucleolus}
defines a unique and supposedly fair cost allocation.
For a formal definiton, let $y\in\bR^P$, and consider the excess vector $\theta(y)$
that contains all coalitional \emph{excess} values
$(y(S)-v(S))_{S\subseteq P}$ in non-decreasing order.
The \emph{nucleolus} is the unique cost allocation
that lexicographically maximizes the excess vector $\theta(y)$
among all $y\in \bR^P$ with $y(P) = v(P)$.

\subsection{Our Results}

We study the cost allocation concept of the nucleolus for cooperative games.
Its computation can be reduced to separating over the polytopes in the classical MPS scheme:
For a proposed cost allocation, we have to find a coalition with smallest excess that is not yet fixed.
The latter requirement results in a linear subspace avoidance constraint.
We show that linear subspace avoidance (LSA) is polynomially equivalent to requring a \emph{non-zero} (NZ) coalition in the following sense:
Given integral values per player, the total sum of these values in our coalition cannot be $0$.   

\begin{restatable*}{proposition}{lsanzequivalence}
	\label{lem:equivalence_subspace_avoidance_non_zero}
	The problems \Psubspaceavoidingminexcess{} and \Pnonzerominexcess{} are polynomially equivalent.
\end{restatable*}

This representation proves to be extremely powerful.
While linear subspace avoidance has not received much attention beyond applications in game theory, 
non-zero constraints have appeared in various fields.
For example, we can make use of previous results about non-zero matroid basis computation \cite{hoersch2024problems}
and non-zero submodular function minimization \cite{goemans1995minimizing}. 
Further, we extend the list of games for which the nucleolus
can be computed in polynomial time:
 
\begin{theorem}
	\label{thm:poly_nuc}
	The nucleolus can be computed in polynomial time for
	\begin{itemize}
		\item b-matching games with $b \leq 2$, if $b=2$ for only constantly many vertices,
		\item b-matching games with $b \leq 2$, if \textsc{Shortest Non-Zero Cycle} can be solved in polynomial time,
		\item arboricity games,
		\item network strength games.
	\end{itemize}
\end{theorem}
The \textsc{Shortest Non-Zero Cycle} problem is closely related to the
\textsc{Shortest Odd Cycle} problem, which is conjectured to be polynomially
solvable \cite{schlotter2025odd}. Therefore, this seems like a promising
approach to compute the nucleolus of general $b$-matching games with $b \leq 2$.

Beyond these applicational results, we study linear subspace avoidance from a more general perspective.
It was known before that linear subspace avoidance constraints strictly increase the complexity of minimum excess problems in general, unless $\mathrm{P}=\mathrm{NP}$ \cite{konemann2020general}.
The reduction relies on a non-monotone game.
We show that the same complexity increase occurs for the class of monotone games, even if $\mathrm{P} = \mathrm{NP}$.
On the positive side, we provide a polynomial-time reduction from \Psubspaceavoidingminexcess{} to \Pminexcess{} that trades the linear subspace avoidance constraints against an arbitrarily small approximation error.
Unfortunately, such an error might propagate exponentially in the MPS scheme.
The nucleolus defines a continuous mapping based on the value function of a coooperative game~\cite{schmeidler1969nucleolus}. 
We show that this mapping is generally unstable, even when restricting to monotone and superadditive games, in the following sense:

\begin{restatable*}{theorem}{nucleolusunstable}
	\label{thm:nucleolus-unstable}
	Let $n \in \bN$ and $\varepsilon > 0$. There are monotone, superadditive
	games $(P,v)$ and $(P,\tilde v)$ with $\abs{P}=O(n)$ such that
	$\abs{v(S) - \tilde v(S)} \leq \varepsilon$
	for all $S\subseteq P$ but the nucleoli of $(P,v)$ and $(P,\tilde v)$
	differ by $2^n\cdot\varepsilon$ for some players.
\end{restatable*}

\subsection{Related Work}

Since its definition by Schmeidler in 1969 \cite{schmeidler1969nucleolus}, 
the nucleolus and its computability have received much attention in the context of cooperative game theory.
The first positive result, i.e., an algorithm that computes the nucleolus in polynomial time, 
was attained by \textcite{littlechild1974simple} for minimum cost spanning tree games on line graphs.
More games followed, e.g., convex games \cite{faigle2001computation} 
and spanning connectivity games \cite{aziz2009wiretapping}.
On the other hand, the first negative result, i.e., a proof of NP-hardness of nucleolus computation, 
was attained by \textcite{faigle1998note} for general minimum cost spanning tree games.
Over the years, both lists have grown and are still extended.

We contribute positive results for arboricity games, network strength games, and certain $b$-matching games.
For both arboricity and network strength games, the nucleolus was previously known to be computable in polynomial time in case of a non-empty core 
due to \textcite{xiao2023arboricity} and \textcite{baiou2020network}, respectively.
We will show that the same is possible for all instances. 
Note that moving from instances with non-empty core
to general instances is often non-trivial \cite{blauth2024cost,konemann2020computing,pashkovich2022computing}.
For $b$-matching games, nucleolus computation was shown to be NP-hard by \textcite{konemann2024complexity},
even for bipartite graphs and $b \leq 3$.
The general case of $b \leq 2$ is still open.
Certain special cases were solved in \cite{konemann2024complexity}.
We will provide a detailed overview over past work in \cref{sec:b-matching}.

Next to results for specific games, there are also structural findings
such as the MPS scheme for nucleolus computation by \textcite{kopelowitz1967computation} and \textcite{maschler1979geometric},
and a framework by \textcite{konemann2020general}
that requires a certain structure of a game.
Part of our work follows this line of research.
We will provide a detailed discussion comparing our work with \cite{konemann2020general} in \cref{subsec:congruency}.
Stability of the nucleolus has not been studied thoroughly, but an
example of \textcite{kumabe2024lipschitz} implies that changing the value
function of a game by a constant may change the nucleolus by $\Omega(\abs{P})$. We show
that this change can even be exponential in $\abs{P}$.

Finally, we want to mention that all our results carry over
to cooperative games defined by a \emph{cost} function (instead of a \emph{value} function), as well as
related solution concepts which also lexicographically maximize the excess vector,
but require a different total value $y(P)$
\cite{blauth2024cost,aziz2010monotone,meir2011subsidies,engevall2004heterogeneous}.

\section{Nucleolus Computation by Constrained Optimization}
\label{sec:separation}

In this section, we study separation over the polytopes of the classical Maschler-Peleg-Shapley (MPS) scheme for nucleolus computation,
which gives rise to the following problem:

\begin{problem}[\Psubspaceavoidingminexcess]
	Given a cooperative game $(P,v)$, a cost allocation $y \in \bR^P$,
	and a linear subspace $L \subsetneq \bR^P$ (described by a basis),
	find a coalition $S\subseteq P$ with $S \notin L$ that minimizes $y(S)-v(S)$.
\end{problem}

When writing $S \notin L$, we identify coalitions with their incidence vectors.
We show that  linear subspace avoidance is equivalent to non-zero constraints: 

\begin{problem}[\Pnonzerominexcess]
	\label{problem:min-excess-non-zero-coalition}
	Given a cooperative game $(P,v)$, a cost allocation $y \in \bR^P$,
	and $a \in \bZ^P$, find a coalition $S\subseteq P$ with
	$a(S)\neq0$ that minimizes $y(S) - v(S)$.
\end{problem}

We will compare this formulation to previous work by \textcite{konemann2020general} that relied on congruency constraints, and show that the latter are strictly harder than linear subspace avoidance (and therefore also than non-zero constraints).

\subsection{MPS Scheme}
\label{subsec:mps}

As shown by \textcite{maschler1979geometric}, the nucleolus can be computed by solving a sequence of linear programs. 
An improved variant which needs at most a linear number of iterations was proposed by \textcite{solymosi1993computing},
see \cref{fig:mps}.

\begin{figure}[ht]
	\centering
	\scalebox{0.8}{
	\begin{tcolorbox}[
		colback=gray!10,
		halign=left,
		rounded corners,
		colframe=gray!50,
		boxrule=1.5pt,
		left=25pt
		]
		\begin{itemize}
			\item[\textbf{Step 1:}] Let $\cS_\text{fixed} \gets \emptyset$.\\[.2cm]
			\item[\textbf{Step 2:}] Compute a pair of optimum primal and dual solutions $((\xi, y), z)$ to \eqref{eq:mps-scheme-lp}, where $\mathrm{span}(\cS_\text{fixed})$ denotes the linear subspace of $\bR^P$ generated by the incidence vectors of coalitions in $\cS_\text{fixed}$.
			\begin{equation}
				\small 
				\begin{alignedat}{3}
					&\max\mathrlap{\ \xi} \\
					&\ \text{s.t.}\quad& y(S)&\ = \ v(S) + \xi^*_S
					\qquad &&\text{for} \ S\in\cS_\text{fixed} \\
					&&y(S) &\ \geq \ v(S) + \xi
					\qquad &&\text{for} \ S\in 2^P\setminus
					\mathrm{span}(\cS_\text{fixed}) \\
					&&y(P)&\ = \ v(P)\\
					&&\xi &\ \in \ \bR \\
					&&y &\ \in \ \bR^P
				\end{alignedat}\label{eq:mps-scheme-lp} \tag{LP}
			\end{equation}
			\item[\textbf{Step 3:}] For each $S\in 2^P\setminus\mathrm{span}(\cS_\text{fixed})$ with $z_S \neq 0$, $\cS_\text{fixed} \gets \cS_\text{fixed} \cup \{S\}$, $\xi^*_S \gets \xi$.\\[.2cm]
			\item[\textbf{Step 4:}] Repeat Steps 2 and 3 until $\mathrm{span}(\cS_\text{fixed})=\bR^P$.
			Return $y$.
		\end{itemize}
	\end{tcolorbox}
	}
	\caption{MPS scheme}
	\label{fig:mps}
\end{figure}

In order to apply this scheme, we need to be able to solve the linear program (\ref{eq:mps-scheme-lp}).
All other steps are trivial. 
If the constraints 
\[
y(S) \ \geq \ v(S) + \xi \qquad \text{for } \ S\in 2^P\setminus \mathrm{span}(\cS_\text{fixed})
\] 
can be separated efficiently, 
then we can also solve (\ref{eq:mps-scheme-lp}) and its dual \cite{grotschel1981ellipsoid}, 
and therefore compute the nucleolus.
This motivates the following well-known result:

\begin{proposition}
	\label{prop:lsa_for_nuc_poly_time}
	If we can solve \Psubspaceavoidingminexcess{} for a cooperative game $(P,v)$  in polynomial time (in $|P|$ and the representation size of $v$), then the nucleolus of $(P,v)$
	can also be computed in polynomial time. 
	\hfill \qed 
\end{proposition}

\subsection{Introducing Non-Zero Constraints}

We will reformulate linear subspace avoidance to a non-zero constraint, 
which we believe easier to handle and think about.
Such a non-zero constraint already appeared in optimization outside of game theory
\cite{hoersch2024problems,kobayashi2017finding,thomas2023packing}.

\lsanzequivalence

\begin{proof}
	For a vector $a \in \bZ^P$ that is not $0$ everywhere, 
	\[
	L \ = \ \mathrm{span}(\{S\subseteq P \ :\  a(S)=0\})
	\] 
	defines a proper linear subspace of $\bR^P$, 
	so finding a non-zero coalition is a special case of linear subspace avoidance, 
	namely the special case of an $(n-1)$-dimensional linear subspace. 
	Transforming $L$ into basis representation can be done in polynomial time
	by Gaussian elimination.
	
	Given a general linear subspace $L$, represented by a basis, this can be expressed by the null space of a matrix $A \in \bQ^{k\times|P|}$, i.e.,
	\[
	L \ = \ \{x \in \bR^P \ : \ \langle a_i \, , \, x \rangle = 0\quad \forall i=1,\ldots,k\} \enspace ,
	\]
	where $a_i$ are the rows of $A$ and $k=\abs{P}-\dim L$. 
	By scaling everything up, we can additionally require $a_i \in \bZ^P$. 
	It then suffices to compute a minimum-excess coalition $S_i \subseteq P$ with $a_i(S_i)\neq 0$ for all
	$i=1,\ldots,k$, and choose the best one.
	The matrix $A$ can be computed in polynomial time by Gaussian elimination.
	\qed
\end{proof}

\subsection{Congruency-Constrained Optimization}
\label{subsec:congruency}

\textcite{konemann2020general} also considered \Psubspaceavoidingminexcess{}
to compute the nucleolus. They proved that it can be reduced to
a congruency-constrained problem. Additionally, if an optimization problem
can be solved via dynamic programming, then so can its congruency-constrained
variant.
We restate the first result here, since it
will be useful for the proof of \cref{thm:b-matching-randomized} in \cref{sec:b-matching}.

\begin{lemma}[\textcite{konemann2020general}]\label{lemma:congruency-min-excess-coalition}
	Let $M$ be a finite set, $f: 2^M \to \bR$, and $L \subseteq \bR^M$
	a linear subspace (given by a basis). Then, $S \subseteq M$ with
	$S \notin L$ that maximizes $f(S)$ can be computed by solving a polynomial
	(in $\abs{M}$) number of problems of the form
	\[
	\max \left\{ f(S) \ : \ a(S) \equiv k \mod p \right\} \enspace ,
	\]
	where $p$ is a polynomially bounded prime, $a \in (\bZ_p)^M$, and
	$k \in \bZ_p$.
	\hfill \qed 
\end{lemma}

While the constraint $a(S) \equiv k \mod p$ may seem similar to the
constraint $a(S) \neq 0$ of \Pnonzerominexcess{}, it turns out that
congruency constraints are strictly harder in general: Given a matroid on
a ground set $E$, $a \in \bZ^E$, and a target number $k \in \bZ$, finding
a basis $B$ of the matroid with $a(B)=k$ is impossible with only a
polynomial number of calls to an independence oracle \cite{doron2026lower}.
This is still true if $a$ and $k$ are polynomially bounded and we allow
for randomized algorithms. Note that this corresponds to the special case
of a congruency-constrained problem where $p$ is much larger than $a$.
In contrast, finding a basis (or an arbitrary independent set)
$B$ with $a(B) \neq 0$ can be done
in oracle polynomial time \cite{hoersch2024problems}. One can even find
such a basis (or independent set) with maximum weight. This fact is the
foundation for our results on games with a matroid structure in \cref{sec:arboricity-and-network-strength}. 
It further shows that, while \Psubspaceavoidingminexcess{}
and \Pnonzerominexcess{} are equivalent, the same is not true for the
congruency-constrained variant of \textcite{konemann2020general}, even when restricting
to monotone games.
For this reason, we focus our attention on the \Pnonzerominexcess{} problem.

\section{Application to Specific Classes of Games}

In this section, we focus on \Psubspaceavoidingminexcess{} for specific classes of games.
Our proofs  rely on the equivalent formulation as \Pnonzerominexcess{}, cf.\ \cref{lem:equivalence_subspace_avoidance_non_zero}.
Recall that if we can solve \Psubspaceavoidingminexcess{} in polynomial time, we can also compute the nucleolus in polynomial time by \cref{prop:lsa_for_nuc_poly_time}.
Our positive results are summarized in \cref{thm:poly_nuc}, which is a consequence of \cref{thm:equiv-matching-problems,thm:lsa_b-matching_constant_2,thm:matroid-games-subspace-avoidance}.

Further, we can recover previous positive results on nucleolus computability:
\Psubspaceavoidingminexcess{} is solvable in polynomial time for \emph{convex games} and \emph{spanning connectivity games}.
For both, polynomial time computability of the nucleolus was known before, but derived differently.
Convex games are characterized by a supermodular value function.
We can solve \Pnonzerominexcess{}
due to a result by
\textcite{goemans1995minimizing} on submodular functions, and the fact
that non-zero coalitions constitute what they call a \emph{parity family}.
Spanning connectivity games \cite{aziz2009power} are based on the graphic matroid, and fall under \cref{thm:matroid-games-subspace-avoidance}.
\cref{fig:table_games} summarizes our findings. 

\begin{figure}
	\newcommand{\newresult}{\textbf{\textcolor{blue!80!black}{this paper}}}
	\newcommand{\knownresult}{known}
	\begin{center}
		\scalebox{.8}
		{
			\onehalfspacing
			\begin{tabular}{!{\vrule width 1.1pt}c!{\vrule width 1.1pt}c|c!{\vrule width 1.1pt}}
				\noalign{\hrule height 1.1pt}
				\rowcolor{gray!20} & \ \textbf{\Psubspaceavoidingminexcess{}} \ & \ \textbf{\textsc{Nucleolus}} \ \\ \noalign{\hrule height 1.1pt}
				\cellcolor{gray!20} Matching \ & \newresult{} & \knownresult{} \cite{konemann2020computing} \\
				\cellcolor{gray!20} $b$-Matching, $b \leq 2$, $b=2$ constantly often \ & \newresult{} & \newresult{} \\
				\cellcolor{gray!20} Arboricity & \newresult{} & \newresult{} \\
				\cellcolor{gray!20} Network Strength & \newresult{} & \newresult{} \\
				\cellcolor{gray!20} Spanning Connectivity  & \newresult{} & \knownresult{} \cite{aziz2009wiretapping} \\
				\cellcolor{gray!20} Convex & \newresult{} & \knownresult{} \cite{faigle2001computation} \\
				\noalign{\hrule height 1.1pt}
			\end{tabular} 
		}
		\caption{Games for which \Psubspaceavoidingminexcess{} and \textsc{Nucleolus} can be solved in polynomial time.
		Results of this paper are highlighted.
		}
		\label{fig:table_games}
	\end{center}
\end{figure}

\pagebreak

\subsection{$b$-Matching Games}
\label{sec:b-matching}

In the following part, we will study \Psubspaceavoidingminexcess{} for $b$-matching
games:

\begin{definition}
	\sloppy
	Let $G=(V,E)$ be an undirected graph with edge
	weights $w \in \mathbb{R}^E$ and vertex values $b \in \mathbb{N}^V$.
	The \emph{$b$-matching game} associated with $(G,w,b)$ is the value game
	$(P,v)$ with $P\coloneqq V$.
	For $S \subseteq V$,
	\[ 
	v(S) \ \coloneqq \ \max\left\{w(M) \ : \ M\subseteq E(G[S]) \ , \ \abs{\delta(v) \cap M} \leq b(v) \ \ \forall v\in V\right\} 
	\]
	is the maximum weight of a $b$-matching in the induced subgraph $G[S]$.
\end{definition}

It is known that computing the nucleolus of $b$-matching games is NP-hard even for
bipartite graphs and $b\leq 3$ \cite{konemann2024complexity}.
Thus, we focus on $b \leq 2$, where the nucleolus
is so far known to be computable in polynomial time for certain instances on
bipartite graphs \cite{konemann2024complexity}.
We show that \Psubspaceavoidingminexcess{} for $b$-matching games with $b \leq 2$
is equivalent to \textsc{Shortest Non-Zero Cycle}:

\begin{problem}[\textsc{Shortest Non-Zero Cycle}]
	\sloppy
	Given an undirected graph \mbox{$G=(V,E)$} with conservative edge costs
	$c \in \bR^E$ and $a \in \bZ^E$, compute a cycle $C$ with $a(C)\neq0$ that minimizes $c(C)$.
\end{problem}

This is closely related to \textsc{Shortest Odd Cycle},
where one asks for a cycle with odd (instead of non-zero) total value.
In fact, both problems are a special case of a non-zero constraint in a commutative group,
for the groups $\bZ$ and $\bZ_2$. Problems with a non-zero constraint in
a group have been studied e.g.\ for shortest path
\cite{kobayashi2017finding,iwata2022finding} and matroid problems \cite{hoersch2024problems}.

\textcite{schlotter2025odd} showed that \textsc{Shortest Odd Cycle} is NP-hard if one additionally
requires the cycle to contain a specified vertex.
Their construction can trivially be adapted to show that the same is true
for \textsc{Shortest Non-Zero Cycle}. However, the complexity of both
problems without this constraint is still open. \textcite{schlotter2025odd}
conjecture that \textsc{Shortest Odd Cycle} can be solved in polynomial time.
Their main reason for this is the fact that it can be solved in randomized
polynomial time for polynomially bounded weights, which is also true for
\textsc{Shortest Non-Zero Cycle}, see \cref{thm:b-matching-randomized,thm:equiv-matching-problems}.
Therefore, it seems likely that \textsc{Shortest Non-Zero Cycle} is
solvable in polynomial time, which would immediately imply an algorithm
to compute the nucleolus of general $b$-matching games with $b\leq 2$ in polynomial time.

Finally, we show that \Psubspaceavoidingminexcess{}
for $b$-matching games with $b\leq 2$ can be solved in some new special cases. 
If the weights and cost allocation are polynomially
bounded, then the problem can be solved in randomized polynomial time.
Additionally, we show that it can be solved in polynomial time if
$b=2$ only for a constant number of vertices. For this class of games,
the nucleolus was previously not known to be computable in polynomial time.

\begin{problem}[\textsc{Maximum Weight Non-Zero Matching}]
	Given an undirected graph $G=(V,E)$ with edge weights $w \in \bR^E$ and
	$a \in \bZ^E$, compute a matching $M$ with $a(M) \neq 0$ that maximizes $w(M)$.
\end{problem}

\begin{theorem}\label{thm:equiv-matching-problems}
	The following problems are polynomially equivalent:
	\begin{enumerate}[label=(\roman*)]
		\item\label{thm-item:equiv-b-matching}
		\Psubspaceavoidingminexcess{} 
		for $b$-matching games with $b\leq2$
		\item\label{thm-item:equiv-non-zero-matching}
		\textsc{Maximum Weight Non-Zero Matching}
		\item\label{thm-item:equiv-non-zero-cycle}
		\textsc{Shortest Non-Zero Cycle}
	\end{enumerate}
	Additionally, this equivalence also holds when restricting to instances
	with polynomially bounded weights. More precisely, we require
	$w$ and $y$ (for problem \ref{thm-item:equiv-b-matching}),
	$w$ (for problem \ref{thm-item:equiv-non-zero-matching}),
	and $c$ (for problem \ref{thm-item:equiv-non-zero-cycle})
	to be polynomially bounded.
\end{theorem}

\begin{figure}[ht]
	\begin{center}
		\begin{subfigure}[t]{0.9\textwidth}
			\begin{center}
				\scalebox{.7}{
					\begin{tikzpicture}[scale=2]
						\node[circle, fill, label={[xshift=-1em, yshift=-2em] \bm{$\tilde{v}_2$}}] (1) at (0,0) {};
						\node[circle, fill, label={[xshift=1em, yshift=-2em] \bm{$v_2$}}] (2) at (1,0) {};
						\node[circle, fill, label={[xshift=-1em, yshift=-.5em] \bm{$\tilde{v}_1$}}] (3) at (0,1) {};
						\node[circle, fill, label={[xshift=1em, yshift=-.5em] \bm{$v_1$}}] (4) at (1,1) {};
						\node[circle, fill, label={[xshift=0em, yshift=-.2em] \bm{$v_e$}}] (5) at (2,.5) {};
						\node[circle, fill, label={[xshift=0em, yshift=-.2em] \bm{$w_e$}}] (6) at (3,.5) {};
						\node[circle, fill, label={[xshift=-1em, yshift=-2em] \bm{$w_2$}}] (7) at (4,0) {};
						\node[circle, fill, label={[xshift=1em, yshift=-2em] \bm{$\tilde{w}_2$}}] (8) at (5,0) {};
						\node[circle, fill, label={[xshift=-1em, yshift=-.5em] \bm{$w_1$}}] (9) at (4,1) {};
						\node[circle, fill, label={[xshift=1em, yshift=-.5em] \bm{$\tilde{w}_1$}}] (10) at (5,1) {};
						
						\draw[line width=1.5mm, darkgreen] (2) -- (1)-- (3) -- (4); 
						\draw[line width=1.5mm, darkgreen] (9) -- (10)-- (8) -- (7); 
						\draw[line width=1.5mm, darkred] (4) -- (5) -- (6) -- (9);
						\draw[line width=1.5mm, darkred, dashed] (2) -- (5) -- (6) -- (7);
						
						\node[darkgreen] () at (-.2, .6) {\bm{$K$}};
						\node[cyan!80!black] () at (-.27, .35) {\bm{$a(v)$}};
						\node[darkgreen] () at (.5, 1.2) {\bm{$\frac{K + y(v)}{2}$}};
						\node[darkgreen] () at (.5, -.2) {\bm{$\frac{K + y(v)}{2}$}};
						
						\node[darkgreen] () at (5.2, .6) {\bm{$K$}};
						\node[cyan!80!black] () at (5.27, .35) {\bm{$a(w)$}};
						\node[darkgreen] () at (4.5, 1.2) {\bm{$\frac{K + y(v)}{2}$}};
						\node[darkgreen] () at (4.5, -.2) {\bm{$\frac{K + y(v)}{2}$}};
						 
						\node[darkred] () at (2.5, .7) {\bm{$K$}}; 
						\node[darkred] () at (1.3, .5) {\bm{$\frac{K + w(e)}{2}$}}; 
						\node[darkred] () at (3.7, .5) {\bm{$\frac{K + w(e)}{2}$}}; 
					\end{tikzpicture}
				}
				\subcaption{Gadgets for $e=\{v,w\}$}
			\end{center}
		\end{subfigure}
		\vspace{.3cm}
		
		\begin{subfigure}[t]{0.3\textwidth}
			\begin{center}
				\scalebox{.7}{
					\begin{tikzpicture}[scale=2]
						\node[circle, fill] (1) at (0,0) {};
						\node[circle, fill] (2) at (1,0) {};
						\node[circle, fill] (3) at (1,1) {};
						\node[circle, fill] (4) at (0,1) {};
						
						\draw[line width=1.5mm, darkgreen] (1) -- (2) 
						node[midway, below] {\bm{$K-w$}} 
						node[midway, above, cyan!80!black] {\bm{$+a$}};
						
						\draw[line width=1.5mm, darkred] (2) -- (3) 
						node[midway, right] {\bm{$w-K$}} 
						node[midway, left, cyan!80!black] {\bm{$-a$}};
						
						\draw[line width=1.5mm, darkgreen] (3) -- (4) 
						node[midway, above] {\bm{$K-w$}}
						node[midway, below, cyan!80!black] {\bm{$+a$}};
						
						\draw[line width=1.5mm, darkred] (4) -- (1) 
						node[midway, left] {\bm{$w-K$}}
						node[midway, right, cyan!80!black] {\bm{$-a$}};
					\end{tikzpicture}
				}
				\subcaption{Augmenting cycle}
			\end{center}
		\end{subfigure}
		\hspace{1cm}
		\begin{subfigure}[t]{0.3\textwidth}
			\begin{center}
				\scalebox{.7}{
					\begin{tikzpicture}[scale=2]
						\node at (0,0.75) {};
						\node at (0,-0.75) {};
						\node[circle, fill] (1) at (0,0) {};
						\node[circle, fill] (2) at (1,0) {};
						\node[circle, fill] (3) at (2,0) {};
						
						\draw[line width=1.5mm, darkgreen] (1) -- (2)
						node[darkgreen, midway, above] {\bm{$K-c(e)$}};
						
						\draw[line width=1.5mm, darkred] (2) -- (3)
						node[darkred, midway, above] () at (1.5, .2) {\bm{$K$}};
						
						\node[cyan!80!black] (a) at (1,-.25) {\bm{$a(e)$}};
					\end{tikzpicture}
				}
				\subcaption{Gadget for $e$}
			\end{center}
		\end{subfigure}
	\end{center}
	\caption{Reductions in the proof of \cref{thm:equiv-matching-problems}}
	\label{fig:gadgets}
\end{figure}
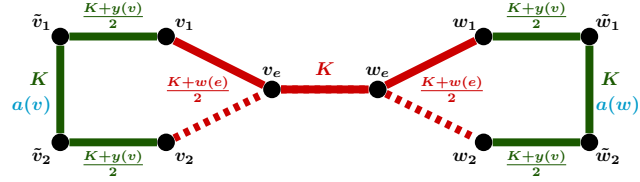
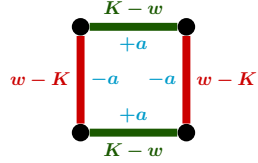
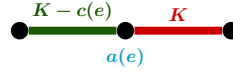

We present the main ideas of the reductions in the following proof sketch. 
A full proof can be found in \cref{appendix:proof-equiv-matching-problems}.

\begin{proof}[sketch]
All reductions work with a sufficiently large constant $K \in \bR_{\geq 0}$ that dominates the objective.
\begin{itemize}
	\item[(a)] Reducing \ref{thm-item:equiv-b-matching} to \ref{thm-item:equiv-non-zero-matching}:\\
	We create node and edge gadgets as illustrated in \cref{fig:gadgets} (a).
	A solution $S$ to \Pnonzerominexcess{} with associated $b$-matching $M$
	corresponds to the matching in the constructed graph that contains all the edges $\{\tilde{v}_1, \tilde{v}_2\}$ for $v \in S$,
	the pairs of edges $\{v_1, \tilde{v}_1\}$, $\{v_2, \tilde{v}_2\}$ for $v \not\in S$, 
	one of the edges $\{v_i, v_e\}$, $i \in \{1,2\}$, for every $e \in M$ and $v \in e$,
	and all the edges $\{v_e, w_e\}$ for $e=\{v,w\} \not\in M$.

	Note that every maximum weight (non-zero) matching in the constructed
	graph contains edges with total weight $\approx K$ in every node and
	edge gadget, which ensures that it corresponds to a coalition in the
	$b$-matching game.
	
	\item[(b)] Reducing \ref{thm-item:equiv-non-zero-matching} to \ref{thm-item:equiv-non-zero-cycle}:\\
	We start by computing a maximum weight matching $\bar M$. If $a(\bar M) \neq 0$,
	then it is an optimum solution, and we are done. Otherwise, $\bar M \Delta M^*$
	can be decomposed into paths and cycles, where $M^*$ is an optimum solution.
	By adding trivial edges, we can restrict to cycles.

	Augmenting along these cycles can only decrease the weight, and at
	least one of them, say $C$, satisfies $a(C \cap \bar M) \neq a(C \cap M^*)$.
	It therefore suffices to find a shortest non-zero cycle in a modified
	graph, see \cref{fig:gadgets} (b).
	The same techniques were used in a slightly different context in \cite{el2025finding}.

	\item[(c)] Reducing \ref{thm-item:equiv-non-zero-cycle} to \ref{thm-item:equiv-b-matching}:\\
	We replace every edge $e$ by a path of length 2 with an additional vertex $v_e$, setting
	$a(v_e) \coloneqq a(e)$ and $a \coloneqq 0$ for all other vertices, as
	shown in \cref{fig:gadgets} (c).
	One of the path edges gets a cost of $K-c(e)$, the other one a cost of $K$.
	By setting $b \equiv 2$ and $y \equiv K$, we
	can ensure that every optimum solution to \Pnonzerominexcess{} for the
	$b$-matching game corresponds to a set of disjoint cycles.
	Since $c$ is conservative, only one of these cycles is non-trivial, and
	therefore an optimum solution to the \textsc{Shortest Non-Zero Cycle} problem.
	\hfill($\square$)
\end{itemize}
\end{proof}

\begin{theorem}
	\label{thm:b-matching-randomized}
	\Psubspaceavoidingminexcess{} for $b$-matching
	games with $b\leq2$ can be solved in randomized polynomial time if
	$w$ and $y$ are polynomially bounded.
\end{theorem}
\begin{proof}
By \cref{thm:equiv-matching-problems}, it suffices to solve \textsc{Maximum Weight Non-Zero Matching}.
By \cref{lem:equivalence_subspace_avoidance_non_zero},
we can transform any non-zero problem into a linear subspace avoiding problem,
which in turn can be reduced to a congruency-constrained problem as in \cref{lemma:congruency-min-excess-coalition}.
Let $(G,w,a,k)$ be an instance of the congruency-constrained version of
\textsc{Maximum Weight Non-Zero Matching}, with $G=(V,E)$, $w \in \mathbb{R}^E$, 
$a \in (\bZ_{p})^E$, and $k \in \bZ_p$ for some polynomially bounded prime $p$. 
We want to find a matching $M$ with $a(M) \equiv k \mod p$ that maximizes $w(M)$.
Let $M^*$ be such an optimum solution.

Let $L\coloneqq 2\sum_{e \in E}\abs{w(e)}+1$ and 
$w'(e) \coloneqq w(e) + L\cdot a(e)$ for all $e \in E$, 
where we view $a(e)$ as an integer in $[0, p-1]$. For every $r \in \bZ$ with
$\abs{r} \leq \sum_{e \in E}\abs{w'(e)}$, we can compute a matching $M_r$
with $w'(M_r)=r$, if it exists, in randomized polynomial
time, because $w'$ is polynomially bounded \cite{mulmuley1987matching}.
Let $r^*\coloneqq w'(M^*)$.
Since $L$ was chosen sufficiently large, we must have $a(M_{r^*})=a(M^*)$,
$w(M_{r^*})=w(M^*)$,
and $M_{r^*}$ is an optimum solution. 
It therefore suffices to check all the previously computed matchings
for feasibility and return the best one.
\qed
\end{proof}

\begin{theorem}
	\label{thm:lsa_b-matching_constant_2}
	\Psubspaceavoidingminexcess{} for $b$-matching
	games with $b\leq 2$ can be solved in polynomial time if $b(v)=2$ only
	for a constant number of vertices $v$.
\end{theorem}
\begin{proof}
We first show that it suffices to solve instances of \textsc{Shortest Non-Zero Cycle}
that have an optimum solution with only a constant number of non-zero
edges:
\begin{claim}
	\Psubspaceavoidingminexcess{} for $b$-matching
	games with $b\leq2$ and $b=2$ for at most $k$ vertices can be reduced to
	\textsc{Shortest Non-Zero Cycle} restricted to instances that have an
	optimum solution with at most $k+2$ non-zero edges.
\end{claim}
\textit{Proof of Claim.}
This follows from the proof of \cref{thm:equiv-matching-problems}. 
In the reduction from \ref{thm-item:equiv-b-matching}
to \ref{thm-item:equiv-non-zero-matching}, any non-zero edge $e$ of the
resulting graph corresponds to a vertex $v$ of the
original graph. If $b(v)=1$, then one of the endpoints of $e$ will lead to a dead end:
This endpoint was called $\tilde v_2$ in the reduction, which is a degree-2 vertex that 
connects to the degree-1 vertex $v_2$. In particular,
if the $b$-matching instance contains at most $k$ vertices $v$
with $b(v)=2$, then any path or cycle in the resulting graph can contain
at most $k+2$ non-zero edges.
In the reduction from \ref{thm-item:equiv-non-zero-matching} to
\ref{thm-item:equiv-non-zero-cycle}, the optimum non-zero
cycle corresponds to an augmenting cycle or path in the graph from
\ref{thm-item:equiv-non-zero-matching}. Therefore, it
contains at most $k+2$ non-zero edges, proving the claim.
\hfill$\diamondsuit$

Let $(G,c,a)$ be an instance of \textsc{Shortest Non-Zero Cycle} with $G=(V,E)$,
$C^*$ an optimum solution,
and $Z \coloneqq \{e \in E : a(e) = 0\}$.
By the claim, we can assume that $E \setminus Z$ has at most constant size, so we can guess $N^* = C^* \setminus Z$.
We compute a minimum-cost $\mathrm{odd}(N^*)$-join $J$ in $(V,Z)$.
The union $J \cup N^*$ consists of disjoint cycles with total cost at most $c(C^*)$
and total non-zero $a$-value. Since $c$ is
conservative, one of these cycles is an optimum solution.
\qed 
\end{proof}

\subsection{Arboricity and Network Strength Games}
\label{sec:arboricity-and-network-strength}

Next, we will study games with an underlying matroid structure, 
and show how to solve \Psubspaceavoidingminexcess{} in polynomial time based on results by \textcite{hoersch2024problems} on non-zero matroid bases. 
We will reduce to the following problem:

\begin{problem}[\textsc{Maximum Weight Non-Zero Independent Set}]
	Given a matroid $(E,\cF)$ (by a polynomial-time independence oracle), 
	weights $w \in \bR^E$, and $a \in \bZ^E$,
	find $F \in \cF$ with $a(F) \neq 0$ that maximizes $w(F)$.
\end{problem}

Specifically, we will study \emph{arboricity games} and \emph{network strength games}:

\begin{definition}
	\sloppy
	Let $G=(V,E)$ be an undirected graph and $P \coloneqq E$. The \mbox{\emph{arboricity game}} associated
	with $G$ is the cost game $(P,c)$, where for $S \subseteq E$,
	\[
	c(S) \ \coloneqq \ \min \left\{k \ : \ \text{$S$ can be covered by $k$ forests} \right\} \enspace .
	\]
	The \emph{network strength game} associated
	with $G$ is the value game $(P,v)$, where for $S \subseteq E$,
	\[
	v(S) \ \coloneqq \ \max \left\{ k \ : \ \text{$S$ contains $k$ disjoint spanning trees} \right\} \enspace .
	\]
\end{definition}

For cost games, the nucleolus is defined analogously to value games, 
but the excess of a coalition $S \subseteq P$ is now given by $c(S) - y(S)$.

\begin{theorem}\label{thm:matroid-games-subspace-avoidance}
	Consider a cooperative game for which \Psubspaceavoidingminexcess{} can be polynomially reduced to \textsc{Maximum Weight Non-Zero Independent Set}.
	Then, \Psubspaceavoidingminexcess{} can be solved in polynomial time.
	
	Especially, this holds for
	\begin{itemize}
		\item[(a)] arboricity games,
		\item[(b)] network strength games.
	\end{itemize}
\end{theorem}

\begin{proof}
	\textcite{hoersch2024problems} show that a maximum weight non-zero
	basis can be computed in polynomial time for every matroid given by a 
	polynomial-time independence oracle. This immediately implies
	that we can also find a maximum weight non-zero independent set:
	Consider the $k$-truncation of $(E,\cF)$, i.e., the restriction to independent
	sets of size at most $k$, for every $1\leq k\leq \abs{E}$.
	This is again a matroid, and a polynomial-time independence oracle
	can easily be constructed. Thus, we can compute a maximum weight 
	non-zero basis in polynomial time.
	Now every independent set of $(E,\cF)$ is
	a basis of some $k$-truncation, so we can simply choose the best
	computed solution. In total, we can solve 
	\textsc{Maximum Weight Non-Zero Independent Set} in polynomial time.
	If \Psubspaceavoidingminexcess{} polynomially reduces to the latter,
	it can therefore also be solved in polynomial time.
	For the specific games, we will use the equivalence established 
	in \cref{lem:equivalence_subspace_avoidance_non_zero}
	and reduce \Pnonzerominexcess{}:
	\begin{itemize}
		\item[(a)] Arboricity games:\\
		Let $(V,E)$ be an undirected graph, $y \in \bR^E$, $a\in \bZ^E$,
		and $k \in \bN$. We will compute a coalition
		$S \subseteq E$ with $c(S)\leq k$ and $a(S)\neq0$ that maximizes $y(S)$.
		Doing this for all $1 \leq k \leq \abs{E}$ suffices to solve \Pnonzerominexcess{}.
		
		Let $(E,\cF)$ be the $k$-fold sum of the graphic matroid on $(V,E)$,
		i.e.,
		\[
		\cF \ \coloneqq \ \left\{F \subseteq E \ : \ F \text{ can be partitioned
			into $k$ forests} \right\} \enspace .
		\]
		By \textcite{nash1966application}, this is a matroid, and it has a
		polynomial-time independence oracle via matroid intersection. 
		It suffices to solve \textsc{Maximum Weight Non-Zero Independent Set} with weights $y$ over $\cF$.
		\item[(b)] Network strength games:\\
		Let $(V,E)$ be an undirected graph, $y \in \bR^E$, $a \in \bZ^E$, and $k \in \bN$.
		Again, by enumerating all relevant values for $k$, 
		\Pnonzerominexcess{} reduces to finding
		a set $S \subseteq E$ with $a(S) \neq 0$ that contains at least $k$ disjoint
		spanning trees (i.e., $v(S)\geq k$) and minimizes $y(S)$.
		Without loss of generality, $a(E)=0$, otherwise add a new edge $e'$
		such that $a(E\cup\{e'\})=0$ and $y(e')$ sufficiently large.
		
		We define the matroid $(E,\cF)$ as before by
		\[
		\cF \ \coloneqq \ \left\{F \subseteq E \ : \ F \text{ can be partitioned
			into $k$ forests}\right\} \enspace .
		\]
		We can assume that the bases of $(E,\cF)$ are exactly the disjoint unions
		of $k$ spanning trees, otherwise there is no coalition $S$ with
		$v(S) \geq k$. Let $(E,\cF^*)$ be the dual of $(E,\cF)$, i.e.,
		\[
		\cF^* \ \coloneqq \ \left\{F \subseteq E \ : \ E \setminus F \text{ contains a basis of } (E,\cF) \right\} \enspace .
		\]
		This is again a matroid with a polynomial-time independence oracle.
		Solving \textsc{Maximum Weight Non-Zero Independent Set} 
		with weights $y$ over $\cF^*$ yields
		an independent set $F \in \cF^*$ with $a(F)\neq0$ and
		$y(F)$ maximum. Then, $E \setminus F$ contains $k$
		disjoint spanning trees, $a(E \setminus F)=-a(F)\neq0$, and
		$y(E \setminus F)$ is minimum.
		So $S=E \setminus F$ is our wanted coalition.
		\qed
	\end{itemize}
\end{proof}

\section{\Psubspaceavoidingminexcess{} for General Games}

In this section, we study the hardness of linear subspace avoidance, that is, the complexity gap between \Psubspaceavoidingminexcess{} and \Pminexcess{}:

\begin{problem}[\Pminexcess]
	\label{def:subspace-avoid-min-excess}
	Given a cooperative game $(P,v)$ and a cost allocation $y \in \bR^P$,
	find a coalition $S\subseteq P$ that minimizes $y(S)-v(S)$.
\end{problem}

It was known before that for non-monotone games, linear subspace avoidance might increase complexity from polynomial-time solvability to NP-hardness \cite{konemann2020general}.
We extend this result to monotone games, using completely new ideas.
Then, we provide a reduction from \Psubspaceavoidingminexcess{} to \Pminexcess{} that incurs an arbitrarily small additive error in the approximation guarantee, which is the best we can hope for.
Finally, we show that any such tiny error might propagate exponentially in the MPS scheme, 
so the reduction does not imply satisfying guarantees for approximate nucleolus computation. 
In fact, this leads to a general instability result for the nucleolus as a mapping of the value function of a cooperative game, even if we restrict to monotone and superadditive games.

\subsection{Hardness of Linear Subspace Avoidance}

For non-monotone games, linear subspace avoidance constraints can increase
the complexity of \Pminexcess{}: Computing a shortest $s_1$-$t_2$-path in
a directed graph is easy, but by forcing the inclusion of an artificial edge $(t_1,s_2)$, subspace-avoidance
constraints can encode the NP-hard \textsc{Directed Two Disjoint Paths} problem
\cite{konemann2020general}.
Note that the value function of this game is not monotone, as it must strictly distinguish feasible paths from $s_1$ to $t_2$.
Further, for monotone games, enforcing a single player is easy:
Simply set its $y$-value to $0$ and add it to the result of \Pminexcess{}.
Thus, we need new ideas for monotone games.

\begin{theorem}\label{thm:hardness_of_subspace_avoidance}
	There exists no algorithm that
	solves \Psubspaceavoidingminexcess{} for a general monotone game $(P,v)$
	with at most $\abs{P}^{O(1)}$ oracle calls
	to an oracle that solves \Pminexcess{}.
\end{theorem}
\begin{proof}
Fix $k \in \bN$ with $k\geq 2$, sets $A, B$ with $\abs{A}=\abs{B}=2k$, and let $P \coloneqq A \dot\cup B$.
For $S^* \subseteq P$ with $\abs{S^* \cap A} = k+1$ and $\abs{S^* \cap B} = k-1$,
and $S\subseteq P$, define
\begin{align*}
\bar v(S) \ &\coloneqq \
\begin{cases}
	2k + \frac{1}{2} \quad&\text{if } \ |S|>2k \text{ or } |S \cap A| = |S \cap B| = k \enspace ,\\
	|S| \quad& \text{otherwise} \enspace,
\end{cases} \\
v_{S^*}(S) \ &\coloneqq \
\begin{cases}
	2k + \frac{1}{6} \quad&\text{if } \ S=S^* \enspace ,\\
	\bar v(S) \quad&\text{otherwise} \enspace.
\end{cases}
\end{align*}
Clearly, $(P, \bar v)$ and $(P, v_{S^*})$ define monotone games.

\begin{claim}
Let $y' \in \bR^P$ and $S'$ be a coalition with minimum excess $y'(S) - \bar v(S)$. 
Let $\cS(y')$ contain all sets $S^*$ with $\abs{S^* \cap A}=k+1$ and $\abs{S^* \cap B}=k-1$
such that $S'$ is not a minimum-excess coalition of $(P, v_{S^*})$.
Then, $\abs{\cS(y')} \leq {2k \choose k+1}$.
\end{claim}

\textit{Proof of Claim.}
Clearly, either $S'$ or $S^*$ is a minimum-excess coalition of $(P,v_{S^*})$.
If $S^*$ is the unique optimum solution,
then swapping an element of $S^*$ against an element outside $S^*$ must always increase the excess.
This yields:
\begin{enumerate}[label=(\roman*)]
	\item $y(p) - y(q) > - \frac{1}{6}$ for all $q \in S^*$ and $p \not\in S^*$ ,
	\item $y(b) - y(a) > \frac{1}{3}$ for all $a \in A \cap S^*$ and $b \in B \setminus S^*$ .
\end{enumerate}
Now assume that $\abs{\cS(y')} > {2k \choose k+1}={2k \choose k-1}$.
Then, there must be $S_1, S_2 \in \cS(y')$ with $S_1 \cap A \neq S_2 \cap A$ and $S_1 \cap B \neq S_2 \cap B$.
Otherwise, $\cS(y')$ would be fixed on either $A$ or $B$, and thereby too small.
Consider $a \in (S_2 \setminus S_1) \cap A$ and $b \in (S_1 \setminus S_2) \cap B$.
They satisfy $y(b)<y(a)+\frac{1}{6}$ by (i) with $S^*=S_1$,
and $y(b) > y(a) + \frac{1}{3}$ by (ii) with $S^*=S_2$, which is
clearly a contradiction, and proves the claim.
\hfill $\diamondsuit$

Now consider \Psubspaceavoidingminexcess.
By \cref{lem:equivalence_subspace_avoidance_non_zero}, it is polynomially equivalent to \Pnonzerominexcess{}.
Let $a_{\vert A} \equiv 1$ and $a_{\vert B} \equiv -1$, and $y \equiv 1$.
Then, the unique coalition $S$ with $a(S) \neq 0$ that minimizes the excess $y(S)-v_{S^*}(S)$ is $S=S^*$.

Consider an algorithm with access to $P$, $a$, and $y$, which can make
at most $\big({2k \choose k+1} - 1\big)$ queries to an oracle for \Pminexcess{}. 
Let the oracle always return a minimum-excess coalition for $(P,\bar v)$. 
By the claim, there are at least
\[
\underbrace{{2k \choose k+1}^2}_{\text{Choices for }S^*} \
- \ \underbrace{ \left( {2k \choose k+1} - 1 \right) }_{\text{Number of oracle calls}}
\;\cdot\; \underbrace{{2k \choose k+1}}_{\text{Bound on }\abs{\cS(y')}} \
= \quad {2k \choose k+1} \ > \ 1
\]
sets $S^*$ such that the oracle results are also correct for $(P,v_{S^*})$.
Since each of these games has a different, unique optimum solution to \Pnonzerominexcess{}, the algorithm cannot return an optimum solution to \Pnonzerominexcess{} in general.
Since $\big({2k \choose k+1} - 1\big)$ is exponential in $k=\Omega(|P|)$, this finishes the proof.
\qed 
\end{proof}

\subsection{Reducing \Psubspaceavoidingminexcess{} to \Pminexcess}

In the following, we show how to approximate \Psubspaceavoidingminexcess{} for monotone games
where \Pminexcess{} can be approximated.
The reduction incurs an arbitrarily small additional error in the approximation guarantee.
The intuitive notion of approximation, i.e., requiring a solution whose
excess is at most $\alpha$ times the optimum excess, does not make much sense,
since the optimum excess might well be zero or negative.
Instead, we define:

\begin{definition}[\bm{$\alpha$}-approximation] \label{def:alpha_approx}
	Let $\cF$ be a family of monotone cooperative games with non-negative value
	functions, and $\alpha\geq1$. An \emph{$\alpha$-approximation
	algorithm} for \Pminexcess{} or \Psubspaceavoidingminexcess{} for $\cF$
	is an algorithm that, given $(P,v) \in \cF,\ y\in\bR^P_{\geq 0}$ and a linear
	subspace $L\subsetneq \bR^P$, computes in polynomial time a feasible solution
	$S \subseteq P$ and a lower value bound  $\lambda \leq v(S)$ such that
	\[
	y\left(S\right)-\lambda \ \leq \ \alpha\cdot y\left(S^*\right) \ - \ v\left(S^*\right) \enspace ,
	\]
	where $S^*$ is an optimum solution.
\end{definition}

For $\alpha=1$, it is clear that $\lambda=v(S)$ and the algorithm computes
a solution with minimum excess. For $\alpha>1$, the computed solution may have larger
excess than the optimum one. In this case, we allow $\lambda<v(S)$, which is
relevant for games where the value function cannot be computed exactly in
polynomial time.
This is e.g.\ the case for maximum independent set games and 
maximum clique games as defined in \cite{deng1999algorithmic}.
For technical reasons, we assume that
$\max_{S\subseteq P}v(S)$ is polynomially bounded in the input size.

\begin{lemma}\label{lem:price-collecting-coalition-restr}
	Let $\cF$ be a family of monotone cooperative games and $\alpha\geq1$
	such that there is an $\alpha$-approximation algorithm for \Pminexcess{} for $\cF$.
	Let \mbox{$(P,v)\in\cF$}, $y \in\bR^P_{\geq0}$, and
	$A,B\subseteq P$ disjoint subsets.
	We can find an $\alpha$-approximate solution to \Pminexcess{} 
	restricted to coalitions $S$ with $A\subseteq S \subseteq P\setminus B$ in polynomial time.
\end{lemma}
\begin{proof}
	Let $(P,v, y, A, B)$ be an instance of the restricted
	minimum-excess problem with $(P,v)\in\cF$.
	For $p \in P$, we define
	\[
	\hat y_p \ \coloneqq \ \begin{cases}
		0 \qquad&\text{if }p\in A \enspace , \\
		\alpha\cdot U+1 \qquad&\text{if }p\in B \enspace , \\
		y_p \qquad&\text{otherwise} \enspace ,
	\end{cases}
	\]
	where $U\geq0$ is an upper bound on $v(P)$.
	Let $(S',\lambda')$ be
	an $\alpha$-approximate solution of \Pminexcess{} for $(P,v,\hat y)$.
	Then, $S' \cap B=\emptyset$ (otherwise $S'$ cannot be an $\alpha$-approximate
	solution by definition of $\hat y$), 
	so $S''\coloneqq S'\cup A$ is a feasible
	solution to the restricted problem.
	By monotonicity of $v$, $\lambda' \leq v(S'')$.
	
	Let $S^*$ be an optimum solution to the restricted problem for $y$. Then,
	\begin{align*}
		y\left(S''\right)-\lambda' \
		&= \ \hat y\left(S'\right) \ + \ y\left(A\right) \ - \ \lambda' \\
		&\leq \ \alpha\cdot \hat y\left(S^*\right) \ - \ v \left(S^* \right) \ + \ y\left(A\right) \\
		&\leq \ \alpha\cdot \big( \hat y\left(S^* \right)+y\left(A\right) \big) \ - \ v\left(S^*\right) \\
		&= \ \alpha\cdot y \left(S^*\right) - v\left(S^*\right) \enspace ,
	\end{align*}
	so $\left( S'',\lambda' \right)$ is an $\alpha$-approximate solution to the
	restricted problem for $y$.
	Note that we used $y \geq 0$ for the second inequality. 
	\qed 
\end{proof}
For general monotone games, \cref{thm:hardness_of_subspace_avoidance}
shows that \Psubspaceavoidingminexcess{} cannot be reduced to
\Pminexcess{} without loss. We present a reduction that incurs an arbitrarily small
additional error:

\begin{theorem}\label{thm:red-subspace-avoiding-coal}
	Let $\cF$ be a family of monotone cooperative games and $\alpha\geq1$
	such that there exists an $\alpha$-approximation algorithm for \Pminexcess{} for $\cF$.
	Fix $\varepsilon > 0$.
	There exists an
	$(\alpha+\varepsilon)$-approximation algorithm
	for \Psubspaceavoidingminexcess{} for $\cF$.
\end{theorem}

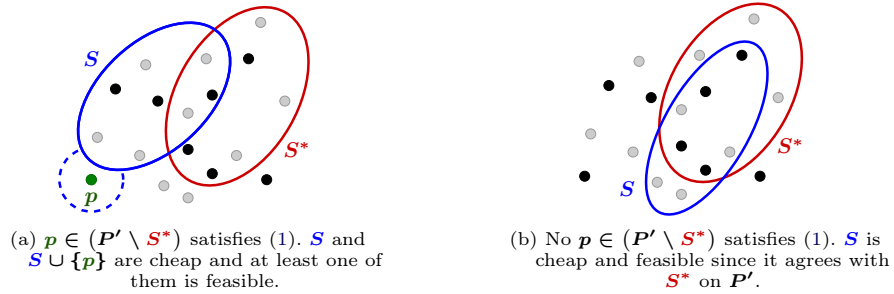
\begin{figure}[ht]
	\newcommand{\ExampleShared}{
		\begin{scope}[every node/.style={circle,fill,draw,scale=0.5}]
			\node (p1) at (0,0) {};
			\node (p2) at (1.1,1.3) {};
			\node (p3) at (2,0.1) {};
			\node (p4) at (2,1.4) {};
			\node (p5) at (2.6,2) {};
			\node (p6) at (0.4, 1.5) {};
			\node (p7) at (1.6, 0.5) {};
			\node (p8) at (2.9, 0) {};
		\end{scope}
		\begin{scope}[every node/.style={circle,fill=gray!40,draw=gray!80,scale=0.5}]
			\node (p'1) at (0.1,0.7) {};
			\node (p'2) at (0.8,0.4) {};
			\node (p'3) at (0.9,1.9) {};
			\node (p'4) at (1.6,1.1) {};
			\node (p'5) at (1.9,2) {};
			\node (p'6) at (2.4, 0.4) {};
			\node (p'7) at (1.6, -0.3) {};
			\node (p'8) at (3.2, 1.3) {};
			\node (p'9) at (2.8, 2.6) {};
			\node (p'10) at (1.2, -0.1) {};
		\end{scope}
		\draw[very thick,darkred,rotate=60] (2.4,-1.4) ellipse (1.6 and 1);
		\node[darkred] at (3.4,0.5) {\bm{$S^*$}};
		\node at (1,-0.4) {};
	}
	\begin{subfigure}[t]{0.45\textwidth}
		\begin{center}
			\scalebox{.8}{
				\begin{tikzpicture}
					\draw[very thick, blue, dashed] (0,0) circle (15pt); 
					\draw[very thick,blue, fill=white, rotate=43] (1.7,0.3) ellipse (1.5 and 0.9);
					\ExampleShared
					\node[circle,draw,fill,scale=0.5,color=green!50!black] (p) at (p1) {};
					\node[color=darkgreen] [below=0 of p] {\bm{$p$}};
					\draw[very thick,blue, rotate=43] (1.7,0.3) ellipse (1.5 and 0.9);
					\node[blue] at (0,2) {\bm{$S$}};
				\end{tikzpicture}
			}
			\subcaption{ \boldmath{$\textcolor{darkgreen}{p} \in \left( P' \setminus \textcolor{darkred}{S^*} \right)$}
				satisfies (\ref{eq:subspace-avoid-pen-bound}). 
				\textcolor{blue}{\bm{$S$}} and \boldmath{$\textcolor{blue}{S} \cup \{ \textcolor{darkgreen}{p} \}$} are cheap 
				and at least one of them is feasible.
				}
		\end{center}
	\end{subfigure}
	\hfill
	\begin{subfigure}[t]{0.45\textwidth}
		\begin{center}
			\scalebox{.8}{
				\begin{tikzpicture}
					\ExampleShared
					\draw[very thick,blue,rotate=60] (1.7,-1.35) ellipse (1.6 and 0.7);
					\node[blue] at (0.7,-0.2) {\bm{$S$}};
				\end{tikzpicture} 
			}
			\subcaption{
				No \boldmath{$p \in \left( P' \setminus \textcolor{darkred}{S^*} \right)$} satisfies (\ref{eq:subspace-avoid-pen-bound}).
				\textcolor{blue}{\bm{$S$}} is cheap and feasible since it agrees with \boldmath{$\textcolor{darkred}{S^*}$} on \bm{$P'$}.
			}
		\end{center}
	\end{subfigure}
	\caption{
		An illustration of the proof of Theorem \ref{thm:red-subspace-avoiding-coal}.
		Players in $\bm{P'}$ are drawn as black points, all others gray.
		\textcolor{darkred}{$\bm{S^*}$} is an optimum solution to the subspace-avoiding problem.
		The solution(s) found by our algorithm are shown in blue.
	}
	\label{fig:red-subspace-avoiding-coal}
\end{figure}

\begin{proof}	
	Let $(P, v, y, L)$ be an instance of \Psubspaceavoidingminexcess{} with $(P,v)\in\cF$ and $y \in \bR^P_{\geq 0}$.
	Let 
	\[
		P' \ \coloneqq \ \left\{ p\in P : {\{p\}}\notin L\right\}
	\]
	and $K\coloneqq \big\lceil \frac{1}{\varepsilon} \big\rceil$.
	Note that $P'\neq\emptyset$.
	Let $S^*$ be an optimum solution. We proceed differently depending on whether there exists $p\in \left( P' \setminus S^* \right)$ with low $y$-value
	\begin{equation}
		y_p \ \leq \ \varepsilon\cdot y(S^*) \enspace . 
		\label{eq:subspace-avoid-pen-bound}
	\end{equation}
	If not, we want to work with a certain \emph{in-set} $I \subseteq S^*$ of bounded size $|I| \leq K$.
	Even though we have to try all possibilities for $p$ and $I$ (because we do not know $S^*$), the procedure takes only polynomial time because $K$ is constant.
	\cref{fig:red-subspace-avoiding-coal} illustrates the following two cases.
	
	\begin{itemize}
		\item[(a)] $p \in \left( P' \setminus S^* \right)$ satisfies (\ref{eq:subspace-avoid-pen-bound}).\\
		Let $(S,\lambda)$ be an $\alpha$-approximate solution to \Pminexcess{} restricted to coalitions that do not contain $p$. 
		This can be computed in polynomial time by Lemma \ref{lem:price-collecting-coalition-restr}. 
		By $p \not\in S^*$,
		\[
			y\left(S\right)\ -\ \lambda \ \leq \  \alpha \cdot y \left(S^* \right) \ -\ v\left(S^*\right) \enspace .
		\]
		By $p \in P'$, we have $S\notin L$ or $(S\cup\{p\})\notin L$. 
		We claim that both solutions are sufficiently cheap, so we can simply return whichever is feasible.
		Since $\lambda$ remains valid for $S \cup \{ p \}$ by monotonicity of $(P,v)$, and $p$ satisfies \eqref{eq:subspace-avoid-pen-bound},
		\begin{align*}
			y\left(S\cup\{p\}\right)\ -\ \lambda
			&\ = \ y_p\ +\ y\left(S\right)\ -\ \lambda \\
			&\ \leq\ (\alpha+\varepsilon)\;\cdot\;y\left(S^*\right) \ - \ v\left(S^*\right) \enspace,
		\end{align*}
		where we used that $S^*$ is an optimum solution to the restricted problem.
		\item[(b)] No $p \in \left( P' \setminus S^* \right)$ satisfies (\ref{eq:subspace-avoid-pen-bound}).\\
		We assume to have found $I \subseteq (S^* \cap P')$ of size $|I| =  \min\{K, \abs{S^* \cap P'}\}$ that maximizes $y(I)$,
		and define the \emph{out-set}
		\[
		O\ \coloneqq \ 
		\begin{cases}
			\big\{ p\in \left( P'\setminus I \right) \ : \ y_p>\min_{q\in I}y_q \big\}
			\qquad &\text{if }\abs{I}=K \enspace , \\[.2cm]
			P'\setminus I &\text{if }\abs{I}<K \enspace .
		\end{cases}
		\]
		We claim that $P' \setminus S^*=O$.
		This is trivial for $\abs{I} < K$ by definition of $I$ and $O$.
		Otherwise, $O \subseteq (P' \setminus S^*)$ because $I$ maximizes $y(I)$ inside $S^* \cap P'$.
		Additionally, $(P' \setminus S^*) \subseteq O$, because for every $p \in (P' \setminus S^*)$, \eqref{eq:subspace-avoid-pen-bound} does not hold, and thus, using $y \geq 0$ for the third inequality,
		\[
		y_p \ > \ \varepsilon \cdot y\left(S^*\right) \ 
		\geq \ \frac{1}{K} \cdot y\left(S^*\right) \ 
		\geq \ \frac{1}{K} \cdot y\left(I \right) \
		\geq \ \min_{q \in I} y_q \enspace .
		\] 
		Let $(S,\lambda)$ be an $\alpha$-approximate solution to \Pminexcess{}
		restricted to coalitions $S$ with $S\cap P'=S^*\cap P'=P' \setminus O$, again using Lemma \ref{lem:price-collecting-coalition-restr}.
		Then, $S^*\notin L$ implies $S\notin L$, so $S$ is feasible. 
		Additionally, $y(S) - \lambda \leq \alpha\cdot y(S^*) - v(S^*)$,
		because $S^*$ is an optimum solution to the restricted problem.
		\qed
	\end{itemize}
\end{proof}

\begin{remark}
	An analogous result holds for cost games. Given a cost game $(P,c)$,
	where $c: 2^P \to \bR$ is monotonically non-decreasing, and $y\in\bR^P_{\geq 0}$, 
	an $\alpha$-approx-imation
	for \Pminexcess{} or \Psubspaceavoidingminexcess{} asks for a
	coalition $S\subseteq P$ such that $c(S) + y(P \setminus S)$ is at
	most $\alpha$ times the minimum possible value. This is sometimes
	called a \textit{prize-collecting} problem, and combining a proof analogous to that
	of \cref{thm:red-subspace-avoiding-coal}
	with known results \cite{blauth2025better,ahmadi2024prize} shows that there is a 1.599-approximation algorithm
	for \textsc{LSA Symmetric Prize-Collecting TSP},
	and a 1.7995-approximation algorithm for
	\textsc{LSA Prize-Collecting Steiner Tree}.
\end{remark}

\subsection{Nucleolus Instability}
\label{sec:nucleolus-unstable}

If \Pminexcess{} can be solved in polynomial time, \cref{thm:red-subspace-avoiding-coal}
implies that \Psubspaceavoidingminexcess{} can be solved with an
arbitrarily small relative error. However, this does not immediately
imply an approximation scheme for computing the nucleolus because
errors may accumulate during the MPS scheme. We will demonstrate that
even small errors when solving \Psubspaceavoidingminexcess{} can result in
large errors in the nucleolus. Thus, solving the separation problem approximately 
by \cref{thm:red-subspace-avoiding-coal} cannot lead to a reasonable
nucleolus approximation in general.

More precisely, we show that changing the value function by some $\varepsilon$
may change the nucleolus of the game by $\varepsilon\cdot 2^{\Omega(\abs{P})}$ per player.
This relates to \cref{thm:red-subspace-avoiding-coal}
in the following way: An oracle for \Psubspaceavoidingminexcess{} with a
relative error of $\varepsilon$ can essentially be viewed as an exact oracle
for a slightly different game, whose value function differs by at most
$\varepsilon\cdot y(S^*)$ from the original value function. Here,
$S^*$ denotes an optimum solution to \Psubspaceavoidingminexcess{}.

\nucleolusunstable

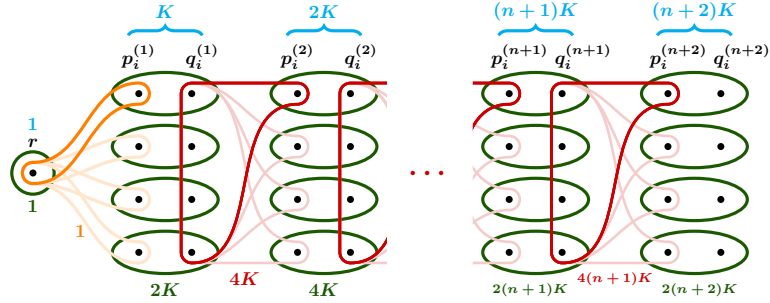
\begin{figure}[ht]
	\begin{center}
		\scalebox{.7}{
			\begin{tikzpicture}[
				player/.style={circle,fill=black,scale=0.4},
				min-exc-set/.style={line width=0.6mm,darkgreen},
				large-exc-set/.style={line width=0.5mm,darkred},
				large-exc-set-transparent/.style={line width=0.5mm,darkred!20},
				nucleolus/.style={line width=.6mm,cyan}
			]
				%
				%
				\node[player] (r) at (0,0) {};
				\node[player] (p11) at (2,1.5) {};
				\node[player] (p12) at (2,0.5) {};
				\node[player] (p13) at (2,-0.5) {};
				\node[player] (p14) at (2,-1.5) {};
				\node[player] (q11) at (3,1.5) {};
				\node[player] (q12) at (3,0.5) {};
				\node[player] (q13) at (3,-0.5) {};
				\node[player] (q14) at (3,-1.5) {};
				\node[player] (p21) at (5,1.5) {};
				\node[player] (p22) at (5,0.5) {};
				\node[player] (p23) at (5,-0.5) {};
				\node[player] (p24) at (5,-1.5) {};
				\node[player] (q21) at (6,1.5) {};
				\node[player] (q22) at (6,0.5) {};
				\node[player] (q23) at (6,-0.5) {};
				\node[player] (q24) at (6,-1.5) {};
				\node[player] (pf1) at (9,1.5) {};
				\node[player] (pf2) at (9,0.5) {};
				\node[player] (pf3) at (9,-0.5) {};
				\node[player] (pf4) at (9,-1.5) {};
				\node[player] (qf1) at (10,1.5) {};
				\node[player] (qf2) at (10,0.5) {};
				\node[player] (qf3) at (10,-0.5) {};
				\node[player] (qf4) at (10,-1.5) {};
				\node[player] (pg1) at (12,1.5) {};
				\node[player] (pg2) at (12,0.5) {};
				\node[player] (pg3) at (12,-0.5) {};
				\node[player] (pg4) at (12,-1.5) {};
				\node[player] (qg1) at (13,1.5) {};
				\node[player] (qg2) at (13,0.5) {};
				\node[player] (qg3) at (13,-0.5) {};
				\node[player] (qg4) at (13,-1.5) {};
				%
				%
				\draw[min-exc-set] (r) circle[radius=0.4];
				\foreach \x in {2.5, 5.5, 9.5, 12.5}{
					\foreach \y in {1.5, 0.5, -0.5, -1.5}{
						\draw[min-exc-set] (\x,\y) ellipse (1 and 0.4);
					}
				}
				%
				%
				\foreach \y in {0.5, -0.5, -1.5}{
				\draw[line width=.6mm, orange!20]
					(-0.2,0)
					to[out=90,in=180] (0, 0.2)
					to[out=0,in=180] (2,\y+0.2)
					to[out=0,in=90] (2.2, \y)
					to[out=270,in=0] (2,\y-0.2)
					to[out=180,in=0] (0, -0.2)
					to[out=180,in=270] (-0.2,0);
				}
				\draw[line width=.6mm, orange]
				(-0.2,0)
				to[out=90,in=180] (0, 0.2)
				to[out=0,in=180] (2,1.5+0.2)
				to[out=0,in=90] (2.2, 1.5)
				to[out=270,in=0] (2,1.5-0.2)
				to[out=180,in=0] (0, -0.2)
				to[out=180,in=270] (-0.2,0);
				%
				%
				\foreach \xleft in {3,6,7,10}{
					\foreach \y in {0.5, -0.5, -1.5}{
						\draw[large-exc-set-transparent]
							(\xleft, -1.7)
							to[out=180,in=270] (\xleft-0.2, -1.5)
							to[out=90,in=270] (\xleft-0.2, 0)
							to[out=90,in=270] (\xleft-0.2, 1.5)
							to[out=90,in=180] (\xleft, 1.7)
							to[out=0,in=180] (\xleft+2, \y+0.2)
							to[out=0,in=90] (\xleft+2.2, \y)
							to[out=270,in=0] (\xleft+2, \y-0.2)
							to[out=180,in=0] (\xleft, -1.7);
					}
					\foreach \y in {0.5, -0.5, -1.5}{
						\draw[large-exc-set]
						(\xleft, -1.7)
						to[out=180,in=270] (\xleft-0.2, -1.5)
						to[out=90,in=270] (\xleft-0.2, 0)
						to[out=90,in=270] (\xleft-0.2, 1.5)
						to[out=90,in=180] (\xleft, 1.7)
						to[out=0,in=180] (\xleft+2, 1.5+0.2)
						to[out=0,in=90] (\xleft+2.2, 1.5)
						to[out=270,in=0] (\xleft+2, 1.5-0.2)
						to[out=180,in=0] (\xleft, -1.7);
					}
				}
				\draw[white, fill] (6.7,-2) rectangle (8.3,2);
				\node[large-exc-set, scale=1.5] at (7.5,0) {\bm{$\cdots$}};
				%
				%
				\node [above=.3 of r] {\bm{$r$}};
				\node[nucleolus] [above=.6 of r] {\bm{$1$}};
		
				\node at ($(p11)+(0,.7)$) {\bm{$p^{(1)}_i$}};
				\node[] at ($(q11)+(.2,.7)$) {\bm{$q^{(1)}_i$}};
				\draw [nucleolus, decorate,decoration={brace,amplitude=5pt}]
				($(p11)+(-.2,1.1)$) -- ($(q11)+(.2,1.1)$) node[midway,yshift=1.3em]{\bm{$K$}};
				
				\node at ($(p21)+(0,.7)$) {\bm{$p^{(2)}_i$}};
				\node[] at ($(q21)+(.2,.7)$) {\bm{$q^{(2)}_i$}};
				\draw [nucleolus, decorate,decoration={brace,amplitude=5pt}]
				($(p21)+(-.2,1.1)$) -- ($(q21)+(.2,1.1)$) node[midway,yshift=1.3em]{\bm{$2K$}};
				
				\node at ($(pf1)+(.2,.7)$) {\bm{$p^{(n+1)}_i$}};
				\node[] at ($(qf1)+(.4,.7)$) {\bm{$q^{(n+1)}_i$}};
				\draw [nucleolus, decorate,decoration={brace,amplitude=5pt}]
				($(pf1)+(-.2,1.1)$) -- ($(qf1)+(.2,1.1)$) node[midway,yshift=1.3em]{\bm{$(n+1)K$}};
		
				\node at ($(pg1)+(.2,.7)$) {\bm{$p^{(n+2)}_i$}};
				\node[] at ($(qg1)+(.4,.7)$) {\bm{$q^{(n+2)}_i$}};
				\draw [nucleolus, decorate,decoration={brace,amplitude=5pt}]
				($(pg1)+(-.2,1.1)$) -- ($(qg1)+(.2,1.1)$) node[midway,yshift=1.3em]{\bm{$(n+2)K$}};
				%
				%
				\node[min-exc-set, below=.3 of r] {\bm{$1$}};
				\node[line width=.6mm, orange] at (.9, -1.1) {\bm{$1$}};
				\node[min-exc-set] at (2.5, -2.2) {\bm{$2K$}};
				\node[large-exc-set] at (4, -2) {\bm{$4K$}};
				\node[min-exc-set] at (5.5, -2.2) {\bm{$4K$}};
				\node[min-exc-set] at (9.4, -2.2) {\scriptsize \bm{$2(n+1)K$}};
				\node[large-exc-set] at (11, -2) {\scriptsize \bm{$4(n+1)K$}};
				\node[min-exc-set] at (12.6, -2.2) {\scriptsize \bm{$2(n+2)K$}};
			\end{tikzpicture}
		}
	\end{center}
	\caption{
		Illustration of the packing games in the proof of \cref{thm:nucleolus-unstable}:
		Vertices illustrate players, and the sets $\cS$ are drawn together with
		their weight $w(\cdot)$.
		For $\tilde w$, the weight of $\{r\}$ is changed to $(1-\varepsilon)$.
		Sets corresponding to zero-excess coalitions (for the nucleoli of both games)
		are drawn in green.
		Nucleolus values of $(P,v)$ are shown in blue. For $(P,\tilde v)$,
		they increase by $2^{l-2}\cdot\varepsilon$ for
		$p^{(l)}_i$, and decrease for $q^{(l)}_i$.
		$K$ may be any sufficiently large constant.
	}
	\label{fig:nucleolus-unstable}
\end{figure}

\begin{proof}
Let $n \in \mathbb{N}$, $\varepsilon > 0$, and
\[
P \ \coloneqq \ \{r\} \ \cup \ \bigcup_{l=1}^{n+2} \; \bigcup_{i=1}^4 \ \left\{\; p^{(l)}_{i} \ , \ q^{(l)}_i \;\right\} 
\]
with $|P| \in \mathcal{O}(n)$.
We define $(P,v)$ and $(P, \tilde v)$ as packing games:
For a family of sets $\cS \subseteq 2^P$ with set weights $w: \cS \to \bR$,
\[
v(S) \ \coloneqq \ \max \ \left\{\sum_{i=1}^k w(S_i) \ : \ S_1,\ldots,S_k \in \cS \ \text{are pairwise disjoint subsets of} \ S \right\}
\]
for $S \subseteq P$, and $\tilde v$ analogously (based on $\tilde w: \cS \to \bR$). 
This clearly results in monotone and superadditive games.
For some fixed $K \geq 2^n\cdot \varepsilon$,
we set $w$ (and therefore implicitly $\cS$) as follows, illustrated in \cref{fig:nucleolus-unstable}:
\begin{align*}
	w(\{r\})\ &=\ 1 \;, \\
	w\left( \left\{ r \, , \, p^{(1)}_i \right\} \right) \ &=\ 1 \qquad&\text{for } &i \in [4] \;, \\
	w\left( \left\{ p^{(l)}_i \, , \, q^{(l)}_i \right\} \right) \ &=\ 2\cdot l\cdot K
		\qquad&\text{for } l \in [n+2] \ , \ &i \in [4] \;, \\
	w\left( \left\{ q^{(l)}_1 \, , \, q^{(l)}_2 \, , \, q^{(l)}_3 \, , \, q^{(l)}_4 \, , \, p^{(l+1)}_i \right\}\right)
		\ &=\ 4\cdot l\cdot K
		\qquad&\text{for } l \in [n+1] \ , \ &i \in [4] \;. 
\end{align*}
We define $\tilde w$ the same way, except $\tilde w(\{r\})=1-\varepsilon$.
Obviously, $|v(S) - \tilde v(S)| \leq \varepsilon$ for all $S \subseteq P$.
We claim that the nucleolus $y$ of $(P,v)$ is given by
\begin{align*}
	y(r)\ &=\ 1 \;,\\
	y\left( p^{(l)}_i \right)\ =\ y\left(q^{(l)}_i\right)\ &=\ l\cdot K \qquad&\text{for } l \in [n+2] \ , \ i \in [4] \;,
\end{align*}
and the nucleolus $\tilde y$ of $(P,\tilde v)$ is given by
\begin{align*}
	\tilde y(r)\ &=\ 1-\varepsilon, \\
	\tilde y \left( p^{(l)}_i \right)\ &=\ l\cdot K + 2^{l-2}\cdot\varepsilon \qquad&\text{for } l \in [n+2] \ , \ i \in [4] \;, \\
	\tilde y\left( q^{(l)}_i \right)\ &=\ l\cdot K - 2^{l-2}\cdot\varepsilon \qquad&\text{for } l \in [n+2] \ , \ i \in [4] \;.
\end{align*}

First, we show that $y$ is in the core of $(P,v)$, and $\tilde y$ is in the core
of $(P,\tilde v)$. It is easy to see that they satisfy group rationality
(note that it suffices to check the coalitions in $\cS$). Additionally,
the coalitions with excess zero with respect to $y$ and $\tilde y$ form
a partition of $P$, so $y(P)=v(P)$ and $\tilde y(P)=\tilde v(P)$ follow
from LP duality: The coalitions with excess zero are a feasible packing, and
$y$ (respectively $\tilde y$) is a feasible solution to the dual of the
packing LP.

This duality argument also shows that any allocation $z \in \bR^P$ in the
core of $(P,v)$ or $(P,\tilde v)$ satisfies
\begin{equation}\label{eq:unstable-core-constraints}
	z\left(p^{(l)}_i\right) + z\left(q^{(l)}_i\right) \ = \ 2\cdot l\cdot K
\end{equation}
for all $l\in [n+2]$ and $i\in[4]$.
Therefore, it suffices to show that, for both $y$ and $\tilde y$, the minimum excess of a coalition
containing $p^{(l)}_i$ but not $q^{(l)}_i$ is equal to the minimum excess
of a coalition containing $q^{(l)}_i$ but not $p^{(l)}_i$: These are the
coalitions that determine how the total value in (\ref{eq:unstable-core-constraints})
is split between the two players.

For both games and allocations, it is easy to see that the minimum-excess coalitions
separating $p^{(l)}_i$ and $q^{(l)}_i$ are among $\left\{p^{(l)}_i\right\}$, $\left\{q^{(l)}_i\right\}$,
and $\left\{r, p^{(1)}_i\right\}$ (for $l=1$)
or $\left\{q^{(l-1)}_1,q^{(l-1)}_2,q^{(l-1)}_3,q^{(l-1)}_4, p^{(l)}_i\right\} \eqqcolon M^{(l)}_i$ (for $l>1$).
Computing these excess values yields
{
	\footnotesize
	\begin{align*}
		y\left(p^{(l)}_i\right) - v\left(\left\{p^{(l)}_i\right\}\right) \ &= \ l \cdot K \qquad \text{for} \ l \in [n+2] \ , \ i \in [4] \enspace ,\\[.3cm]
		y\left(q^{(l)}_i\right) - v\left(\left\{q^{(l)}_i\right\}\right) \ &= \ l \cdot K \qquad \text{for} \ l \in [n+2] \ , \ i \in [4] \enspace , \\[.3cm]
		y\left(\left\{r, p^{(1)}_i\right\}\right) - v\left(\left\{r, p^{(1)}_i\right\}\right) \ &= \ (1+K)-1 \\
		&= \ l\cdot K \qquad \text{for} \ l =1 \ , \ i \in [4] \enspace ,\\[.3cm]
		y\left(M^{(l)}_i\right)-v\left(M^{(l)}_i\right) \ &= \ 4\cdot(l-1)\cdot K + l\cdot K - 4\cdot (l-1)\cdot K \\
		&= \ l\cdot K \qquad \text{for} \ l \in \left( [n+2] \setminus \{1\} \right) \ , \ i \in [4] \enspace .
	\end{align*}
}

Since these excess values are balanced, it follows that $y$ is the
nucleolus of $(P,v)$. For $\tilde y$ and $(P,\tilde v)$, we get

{
	\footnotesize
	\begin{align*}
		\tilde y\left(p^{(l)}_i\right) - \tilde v\left(\left\{p^{(l)}_i\right\}\right) \ &= \ l \cdot K + 2^{l-2} \cdot \varepsilon \qquad \text{for} \ l \in [n+2] \ , \ i \in [4] \enspace ,\\[.3cm]
		\tilde y\left(q^{(l)}_i\right) - \tilde v\left(\left\{q^{(l)}_i\right\}\right) \ &= \ l \cdot K - 2^{l-2} \cdot \varepsilon \qquad \text{for} \ l \in [n+2] \ , \ i \in [4] \enspace , \\[.3cm]
		\tilde y\left(\left\{r, p^{(1)}_i\right\}\right) - \tilde v\left(\left\{r, p^{(1)}_i\right\}\right) \ &= \ \left( 1-\varepsilon+K+\frac{\varepsilon}{2} \right) - 1 \\
		&= \ l\cdot K - 2^{l-2} \cdot \varepsilon \qquad \text{for} \ l =1 \ , \ i \in [4] \enspace ,\\[.3cm]
		\tilde y\left(M^{(l)}_i\right) - \tilde v\left(M^{(l)}_i\right) \ &= \ 4\cdot \left( (l-1)\cdot K - 2^{l-3}\cdot\varepsilon \right)
		+ \left( l \cdot K + 2^{l-2}\cdot\varepsilon \right) \\
		&= \ l\cdot K - 2^{l-2}\cdot\varepsilon \qquad \text{for} \ l \in \left( [n+2] \setminus \{1\} \right) \ , \ i \in [4] \enspace .
	\end{align*}
}

This is the same as for $y$ and $(P,v)$ except for the error terms that
depend on $\varepsilon$. Note that $\left\{p^{(l)}_i\right\}$ is no longer a minimum-excess
coalition, but the excess of the other two coalitions is still balanced.
This shows that $\tilde y$ is the nucleolus of $(P,\tilde v)$.

Since $y$ and $\tilde y$ differ by $2^n\cdot\varepsilon$ for
$p^{(n+2)}_i$ and $q^{(n+2)}_i$ ($i \in [4]$), this proves the theorem.
\qed
\end{proof}

\section{Conclusion}
\label{sec:future_research}

We have introduced the \Pnonzerominexcess{} problem and shown that it can
be used to compute the nucleolus. This relates nucleolus computation
to results on non-zero-constrained optimization.
We demonstrated the power of this approach by applying it
to arboricity and network strength games. This strengthens previous
results on these games (which require a non-empty core), while
at the same time being much simpler. We believe that computing the nucleolus
for more classes of games by solving \Pnonzerominexcess{} is a promising
direction for further research.

For $b$-matching games with $b\leq 2$, we showed that \Psubspaceavoidingminexcess{}
is equivalent to the \textsc{Shortest Non-Zero Cycle} problem. The complexity
of this, and of the related \textsc{Shortest Odd Cycle} problem, are still
unknown. Since there is already some interest in this from different
areas, these are important open problems.
While the general case is still open, we showed that the nucleolus of
$b$-matching games with $b\leq 2$ can be computed if $b=2$ only for
a constant number of vertices. This also gives a new proof that the
nucleolus for matching games can be computed in polynomial time. For
bounded weights and cost allocations, we showed that \Psubspaceavoidingminexcess{}
can be solved in randomized polynomial time. However, turning this into
an algorithm to compute the nucleolus for instances with bounded weights
would require a strong bound on the bit complexity of the $b$-matching
nucleolus. We are not aware of any such bound.

We also provided insights into \Psubspaceavoidingminexcess{} for general
monotone games,
including a reduction to \Pminexcess{} that incurs an arbitrarily small error,
which is best possible.
The nucleolus, however, was shown to be very unstable under slight changes in the value function.

\printbibliography

\newpage
\appendix
\renewcommand{\theHsection}{appendix.\thesection}

\section{Proof of \cref{thm:equiv-matching-problems}}
\label{appendix:proof-equiv-matching-problems}

	We will present polynomial reductions from \ref{thm-item:equiv-b-matching} to \ref{thm-item:equiv-non-zero-matching} to \ref{thm-item:equiv-non-zero-cycle} to \ref{thm-item:equiv-b-matching}.
	All reductions work the same way when restricting to instances with polynomially bounded weights, cost allocations, and costs.
	\begin{itemize}
		\item[(a)] Reducing \ref{thm-item:equiv-b-matching} to \ref{thm-item:equiv-non-zero-matching}:\\
		Let $(G,w,b, a, y)$ be an instance of \Psubspaceavoidingminexcess{}
		for $b$-matching games with $b\leq2$. 
		Recall that for $G=(V,E)$, we have $w \in \bR^E$, $b \in \{1,2\}^V$, $a \in \bZ^V$, and $y \in \bR_{\geq 0}^V$.
		We will polynomially reduce this problem to \textsc{Maximum Weight Non-Zero Matching} by defining an appropriate instance $(G',w',a')$.
		Let $K \coloneqq 2 \cdot \left( \sum_{e \in E} \abs{w(e)} + 2 \cdot \sum_{v \in V} \abs{y(v)} \right) + 1$.
		For every vertex $v \in V$, we add four vertices $v_1, \tilde{v}_1, v_2, \tilde{v}_2$ 
		and three edges $\{v_1, \tilde{v}_1\}$, $\{v_2, \tilde{v}_2\}$ and $\{\tilde{v}_1, \tilde{v}_2\}$ to $G'$,
		and set $w'(\{v_1, \tilde{v}_1\}) = w'(\{v_2, \tilde{v}_2\})=\frac{K+y(v)}{2}$ and $w'(\{\tilde{v}_1, \tilde{v}_2\})=K$.
		We call this the \emph{node gadget} of $v$.
		For every edge $e=\{v,w\} \in E$, we add two vertices $v_{e}$, $w_{e}$, and up to five edges to $G'$:
		$\{v_e, w_e\}$ and $ \{x_i, x_e\}$ for $x \in e$ and $i \leq b(x)$, 
		and set $w'(\{v_e, w_e\}) = K$ and $w'(\{x_i, x_e\}) = \frac{K+w(e)}{2}$.
		We call this the \emph{edge gadget} of $e$.
		Finally, we set $a'(\{\tilde{v}_1, \tilde{v}_2\})=a(v)$ for all $v \in V$, and $a' \equiv 0$ otherwise.
		See \cref{fig:gadgets} (a).
		
		Consider a coalition $S \subseteq V$ with associated $b$-matching $M \subseteq E(G[S])$ with $v(S)=w(M)$. 
		We can define a matching $M'$ in $G'$ by all the edges $\{\tilde{v}_1, \tilde{v}_2\}$ for $v \in S$,
		the pairs of edges $\{v_1, \tilde{v}_1\}$, $\{v_2, \tilde{v}_2\}$ for $v \in V \setminus S$, 
		one of the edges $\{v_i, v_e\}$, $i \in \{1,2\}$, for every $e \in M$ and $v \in e$,
		and all the edges $\{v_e, w_e\}$ for $e=\{v,w\} \in E \setminus M$.
		This matching has weight 
		\[
		w'(M') \ = \ K \cdot (|V| + |E|) + y(V) + v(S) - y(S)
		\]
		and $a'(M') = a(S)$.
		Note that every maximum weight matching in $G'$ will collect weight at least $K$ in every node and edge gadget,
		and this property remains to hold for non-zero maximum weight matchings.
		Especially, for every $e=\{v,w\} \in E$, it will either contain two edges $\{v_i, v_e\}$ and $\{w_j, w_e\}$ for some $i,j \in \{1,2\}$ or none.
		Further, if the matching contains some edge $\{v_i, v_e\}$ for $i \in \{1,2\}$ and $e \in E$, it must also contain the edge $\{\tilde{v}_1, \tilde{v}_2\}$, and thereby collect its $a'$-value of $a(v)$.
		This implies that any maximum weight (non-zero) matching $M'$ in $G'$ yields a coalition $S \coloneqq \{ v \in V : \{\tilde{v}_1, \tilde{v}_2\} \in M' \}$ and a corresponding feasible $b$-matching
		\[
		M \ \coloneqq \ \big\{ e \in E \ : \ \{v_i, v_e\} \in M' \text{ for } v \in e \text{ and some } i \in \{1,2\} \big\} \ \subseteq \ E(G[S])
		\] 
		with 
		\begin{align*}
			w'(M') \ &= \ K \cdot (|V| + |E|) + y(V) + w(M) - y(S)  \\
			&\leq \ K \cdot (|V| + |E|) + y(V) + v(S) - y(S)
		\end{align*}
		and $a'(M') = a(S)$.
		Thus, solving the instance $(G', w', a')$ of \textsc{Maximum Weight Non-Zero Matching} suffices to solve the instance $(G, w, b, a, y)$ of \Psubspaceavoidingminexcess{}.
		\item[(b)] Reducing \ref{thm-item:equiv-non-zero-matching} to \ref{thm-item:equiv-non-zero-cycle}:\\
		Let $(G, w, a)$ be an instance of \textsc{Maximum Weight Non-Zero Matching}.
		Recall that for $G=(V,E)$, we have $w \in \mathbb{R}^E$ and $a \in \mathbb{Z}^E$.
		We will polynomially reduce this problem to \textsc{Shortest Non-Zero Cycle}.
		The same techniques were used in a slightly different context in \cite{el2025finding}.
		By adding a new vertex to $G$ if $\abs{V}$ is odd, and adding trivial edges
		with weight and $a$-value 0 between all pairs of vertices, we can
		ensure that any matching in $G$ can be completed to a perfect matching
		without changing its weight or $a$-value.
		We will thus assume without loss of generality that $G$ already has this property.
		First, we compute a maximum weight matching $\bar M$ in $G$ in polynomial time \cite{edmonds1965maximum}. 
		If $a(\bar M)\neq0$, then $\bar M$ is an optimum solution. 
		Otherwise, let $M^*$ be an optimum solution, i.e., a matching with $a(M^*)\neq0$ and maximum weight.
		As mentioned above, we can assume that $\bar M$ and $M^*$ are perfect matchings.
		Thus, the symmetric difference $\bar M \Delta M^*$ can be decomposed into disjoint $\bar M$-$M^*$-alternating cycles. 
		For each such cycle $C$, \mbox{$w(C\cap \bar M)\geq w(C \cap M^*)$} since $\bar M$ has maximum weight. 
		Additionally, at least one of these cycles $C$ satisfies $a(C \cap \bar M)\neq a(C \cap M^*)$.
		Therefore, a non-zero matching with weight (at least) $w(M^*)$ can be obtained by computing an $\bar M$-augmenting cycle $C$ with $a(C\cap\bar M) \neq a(C \setminus \bar M)$ that minimizes $w(C \cap \bar M)-w(C \setminus \bar M)$.
		Let $K \coloneqq 2\sum_{e \in E}\abs{w(e)}+1$, and for $e \in E$,
		\[
		c(e) \ \coloneqq \ \begin{cases}
			w(e) - K \qquad&\text{if }e \in \bar M \enspace ,\\
			K-w(e) \qquad&\text{if }e \notin \bar M \enspace ,
		\end{cases}
		\]
		and
		\[
		a'(e) \ \coloneqq \ \begin{cases}
			-a(e) \qquad&\text{if }e \in \bar M \enspace ,\\
			a(e) \qquad&\text{if }e \notin \bar M \enspace .
		\end{cases}
		\]
		$(G, c, a')$ defines an instance of \textsc{Shortest Non-Zero Cycle}.
		The construction is illustrated in \cref{fig:gadgets} (b).
		Observe that we have a one-to-one correspondence between non-zero $\bar M$-alternating cycles, and non-zero cycles with cost at most $\frac{K}{2}$. 
		Further, $c$ is conservative, since $\bar M$ has maximum weight.
		Thus, solving the instance $(G,c,a')$ of \textsc{Shortest Non-Zero Cycle} suffices to solve the instance $(G, w, a)$ of \textsc{Maximum Weight Non-Zero Matching}.
		\item[(c)] Reducing \ref{thm-item:equiv-non-zero-cycle} to \ref{thm-item:equiv-b-matching}:\\
		Let $(G,c,a)$ be an instance of \textsc{Shortest Non-Zero Cycle}.
		Recall that for $G=(V,E)$, we have conservative costs $c \in \mathbb{R}^E$, and $a \in \mathbb{Z}^E$.
		We will polynomially reduce this problem to \Psubspaceavoidingminexcess{}
		for a $b$-matching game with $b\leq2$.
		Let $G'=(V',E')$ result from $G$ by subdividing each edge $e \in E$ into a path of two edges $e_1, e_2$ with new middle vertex $v_e$.
		We set $w(e_1)=K-c(e)$, $w(e_2)=K$, 
		and $a'(v_e)=a(e)$.
		Let $K\coloneqq 2\sum_{e \in E}\abs{c(e)}+1$, $y \equiv K$, $b \equiv 2$ on $V'$,
		and $a' | _V \equiv 0$ on the original vertex set $V$.
		$(G', w, b, a', y)$ defines an instance of \Psubspaceavoidingminexcess{}
		for a $b$-matching game with \mbox{$b\leq2$}.
		The construction is illustrated in \cref{fig:gadgets} (c). 
		A non-zero cycle $C$ in $G$ corresponds to the coalition $S \coloneqq V(C) \cup \{v_e : e \in E(C)\}$ with $a'(S)=a(C)\neq0$ and excess
		\[
		y(S)-v(S) \ \leq \ 2\cdot\abs{V(C)}\cdot K - (2\cdot\abs{V(C)}\cdot K - c(C)) \ = \ c(C) \enspace .
		\]
		For any coalition $S$ of the $b$-matching game with excess at most $\frac{K}{2}$, 
		the corresponding $b$-matching clearly contains exactly two edges incident to each vertex in $S$. 
		In particular, it is a set of disjoint cycles in $G$ with total cost $y(S)-v(S)$. 
		If $S$ is a non-zero coalition with minimum excess, then one of these cycles $C$ satisfies $a(C)\neq0$,
		and all other cycles have cost 0 because $c$ is conservative. 
		It follows that $C$ is a shortest non-zero cycle. 
		Thus, solving the $b$-matching game instance $(G', w, b, a', y)$ of 
		\Psubspaceavoidingminexcess{} suffices to solve the instance $(G, c, a)$ of \textsc{Shortest Non-Zero Cycle}.
		\qed 
	\end{itemize}

\end{document}